%% file: main.tex
\documentclass[sigconf,balance=false]{acmart}
\usepackage{popets}

\setcopyright{popets}
\copyrightyear{YYYY}

\acmYear{YYYY}
\acmVolume{YYYY}
\acmNumber{X}
\acmDOI{XXXXXXX.XXXXXXX}
\acmISBN{}
\acmConference{Proceedings on Privacy Enhancing Technologies}
\usepackage{enumitem}
\usepackage{algorithm}
\usepackage{algpseudocode}
\usepackage{amsmath}
\usepackage{amsthm}
\usepackage{multirow}
\usepackage{subcaption}
\usepackage[most]{tcolorbox}

\input{newcommands}

\begin{document}

\title[Differentially Private RAG]{Differentially Private Retrieval-Augmented Generation}


\author{Tingting Tang}
\affiliation{%
  \institution{University of Southern California}
  \city{Los Angeles}
  \state{California}
  \country{USA}}
\email{tangting@usc.edu}

\author{James Flemings}
\affiliation{%
  \institution{University of Southern California}
  \city{Los Angeles}
  \state{California}
  \country{USA}}
\email{jamesf17@usc.edu}

\author{Yongqin Wang}
\affiliation{%
 \institution{University of Southern California}
  \city{Los Angeles}
  \state{California}
  \country{USA}}
\email{yongqin@usc.edu}

\author{Murali Annavaram}
\affiliation{%
 \institution{University of Southern California}
 \city{Los Angeles}
 \state{California}
 \country{USA}}
\email{annavara@usc.edu}


\renewcommand{\shortauthors}{Tang et al.}

\input{body/0_abstract}


\maketitle

\input{body/1_introduction}
\input{body/2_prelim}
\input{body/3_problemsettings}
\input{body/4_method}
\input{body/5_privacyanalysis}
\input{body/6_experiment}
\input{body/7_relatedwork}
\input{body/8_conclusion}


\section*{Acknowledgment}
The authors used Grammarly and ChatGPT4.o to detect and correct grammatical errors in this paper.

\bibliographystyle{ACM-Reference-Format}
\bibliography{references}

\input{body/9_appendix}









\end{document}

%% file: newcommands.tex
\theoremstyle{plain}
\newtheorem{theorem}{Theorem}[section]

\newtheorem{lemma}[theorem]{Lemma}

\theoremstyle{definition}
\newtheorem{definition}[theorem]{Definition}

\theoremstyle{remark}
\newtheorem{remark}[theorem]{Remark}

\newcommand{\method}{\texttt{DP-KSA}}

%% file: body/0_abstract.tex
\begin{abstract}
  Retrieval-augmented generation (RAG) is a widely used framework for reducing hallucinations in large language models (LLMs) on domain-specific tasks by retrieving relevant documents from a database to support accurate responses. However, when the database contains sensitive corpora, such as medical records or legal documents, RAG poses serious privacy risks by potentially exposing private information through its outputs. Prior work has demonstrated that one can practically craft adversarial prompts that force an LLM to regurgitate the augmented contexts. A promising direction is to integrate differential privacy (DP), a privacy notion that offers strong formal guarantees, into RAG systems. However, naively applying DP mechanisms into existing systems often leads to significant utility degradation. Particularly for RAG systems, DP can reduce the usefulness of the augmented contexts leading to increase risk of hallucination from the LLMs. Motivated by these challenges, we present \method{}, a novel privacy-preserving RAG algorithm that integrates DP using the propose-test-release (PTR) paradigm. \method{} follows from a key observation that most question-answering (QA) queries can be sufficiently answered with a few keywords. Hence, \method{} first obtains an ensemble of relevant contexts, each of which will be used to generate a response from an LLM. We utilize these responses to obtain the most frequent keywords in a differentially private manner. Lastly, the keywords are augmented into the prompt for the final output. This approach effectively compresses the semantic space while preserving both utility and privacy. We formally show that \method{} provides formal DP guarantees on the generated output with respect to the RAG database. We evaluate \method{} on two QA benchmarks using three instruction-tuned LLMs, and our empirical results demonstrate that \method{} achieves a strong privacy-utility tradeoff.
\end{abstract}

%% file: body/1_introduction.tex
\section{Introduction}
\label{sec:intro}
Large language models (LLMs) encode copious factual knowledge in their parameters through pre-training on internet-scale data \cite{bert}. Hence, LLMs leverage its knowledge when prompted to accurately answer queries \cite{petroni2019language, radford2019language}. However, the knowledge that the LLM was pre-trained on (1) may not be precisely accessed and utilized and (2) can eventually become outdated. This can lead to discrepancies between generated content by the LLM and verifiable real-world facts, known as \textit{factuality hallucinations} \cite{huang2025survey}. 

\input{body/problem_setup_rag}

A widely-adopted approach to mitigate hallucinations is to update the LLMs' knowledge base by augmenting prompts with external knowledge retrieved from a database, known as retrieval-augmented generation (RAG) \citep{rag, selfrag}. This involves retrieving the top documents from an external database that are semantically relevant (similar) to a user query. However, a major concern is that the external database can contain highly-sensitive information, such as Personable Identifiable Information (PII). For example, healthcare providers can leverage internal medical records to offer accurate diagnoses and tailored care recommendations, while law firms may rely on their legal case repositories to support clients in conducting legal research and preparing documentation. Several have found that RAG on a sensitive corpus can leak private information about individual documents in the corpus \citep{ragattack1, ragattack2, ragattack3}. An attacker can design a specific prompt to a RAG system and use the output to reveal the information retrieved from the sensitive database. 

\input{body/overview}

In this work, we focus on preserving the privacy of the external dataset with the following problem setup shown in Figure \ref{fig:problem_setting}. A RAG system receives a query $\mathbf{x}$. Our RAG system contains a private external dataset $D$, a retriever model $R$ that will retrieve the top documents relevant to the query $\mathbf{x}$, and a generator model $F$ that will use the top retrieved documents and the query to generate a response $\mathbf{y}$. Our setup assumes the adversary only has query access to the RAG system to obtain responses. Hence, because the response can contain information about the retrieved documents, the adversary can adversarially craft prompts such that the response reveals information about the external database. Therefore, our goal is to design a privacy-preserving RAG system to protect private, sensitive information against such attacks. 

One promising privacy notion is Differential Privacy (DP) \citep{dp}, which provides a mathematical guarantee that each individual in a database has limited affect on the outcome of a randomized algorithm. In the context of RAG, each document in the RAG database has limited influence on the output of the LLM. Although DP is the standard privacy safeguard for applications utilizing information-sensitive data, it tends to result in substantial utility degradation. To highlight the challenges that privately generating text presents, note that the response from the generator model $F$ iteratively samples the next token $y_t$ until a stop criteria is met, such as maximum token length reached. Hence, the range of possible values that $y_t$ can take is equal to the vocabulary of the generator's tokenizer, which can be upwards of $50,000$. And if the maximum token length is $10$, then the output space of possible responses is at most $50,000^{10}$. Hence, naively applying standard DP mechanisms, such as additive Gaussian noise, at every token generation could destroy any meaningful utility from the generated tokens. 

Thus, the goal our work is the following: 
\begin{quote}
    \emph{How can we integrate differential privacy into practical RAG systems while preserving their utility?}
\end{quote}

To address the above question, we propose \method{}, a privacy-preserving algorithm based on the propose-test-release (PTR) \citep{ptr} and subsample-and-aggregate \cite{sampleandagg} paradigm for private text generation illustrated in Figure \ref{fig:method_overview}. \method{} derives from a key observation that most queries from question-answering datasets can be answered sufficiently with just a few keywords. Hence, we convert the output space of responses for question-answering tasks into a keyword subspace, then perform differentially privacy to extract the keywords. This keyword subspace mostly preserves the relative semantic representation of the original response space, i.e. it maintains the utility of the LLM while operating in a low-dimensional approximation of the entire sentence space. To achieve this, a retriever model will first retrieve the top-$N$ documents that are most relevant to the query. Next, we partition the documents into prompts, each prompt containing a retrieved document with the corresponding query, and feed these prompts into the generator model to obtain an ensemble of responses. Then \method{} transforms the generated model responses into keywords that preserve relative semantic meaning and privately extracts the frequently occurring keywords in the responses via the propose-test-release paradigm. Finally, the privately obtained keywords are augmented to a prompt containing just the query, which is then fed to the LLM to generate the final output. 

We summarize the \textbf{contributions of our work} as follows:
\begin{enumerate}
    \item We introduce \method{}, a simple RAG framework that preserves the privacy of the external database by generating a response for each of the retrieved documents from the retriever model, then extracting a small set of keywords from the ensemble of responses.
    \item We formally show that the final outputs generated by \method{} can achieve differential privacy, a strong privacy notion that gives a probable guarantee of the privacy leakage of the external database. 
    \item We experimentally demonstrate that \method{} can achieve strong privacy guarantees while preserving utility. In particular, we experimentally evaluated \method{} on two standard benchmarking question-answering datasets with modern instruction-tuned LLMs, Qwen 2.5 \citep{qwen25} and Llama 3 \citep{llama3}.
    \item We provide comprehensive ablation studies for important hyperparameters of \method{} to demonstrate the robustness of \method{} as well as provide insights into the inner workings of \method{}.
\end{enumerate}

%% file: body/problem_setup_rag.tex
\begin{figure}[!t]
  \centering
    \includegraphics[width=\columnwidth]{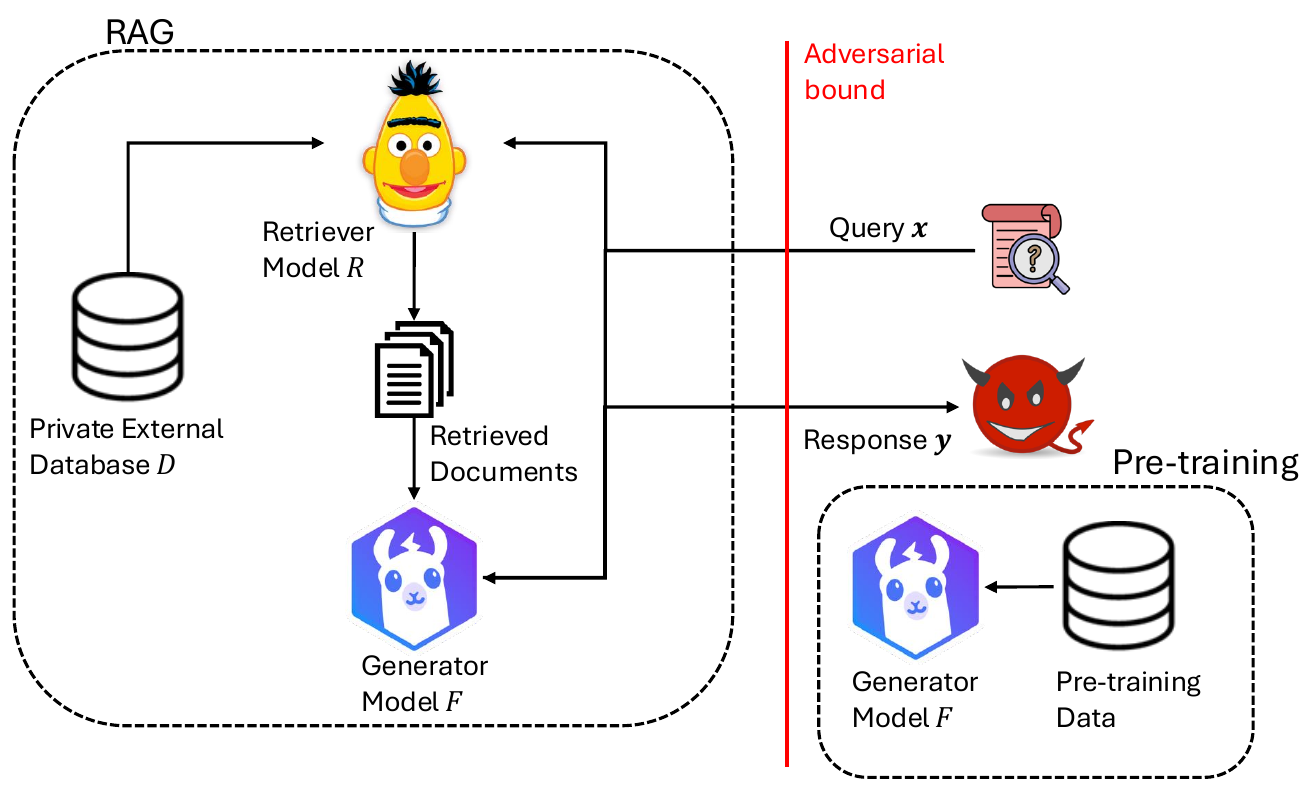}
    \caption{Overview of the DP RAG problem setting. The adversarial bound illustrates what capabilities the adversary has. In this case, the adversary can only query the RAG system and access the answer. The generator model has been pre-trained on publicly available data, which the adversary has access to. However, we are not concerned about preserving the privacy of the pre-training data.}
  \label{fig:problem_setting}
\end{figure}

%% file: body/overview.tex
\begin{figure*}[!t]
  \centering
    \includegraphics[width=\linewidth]{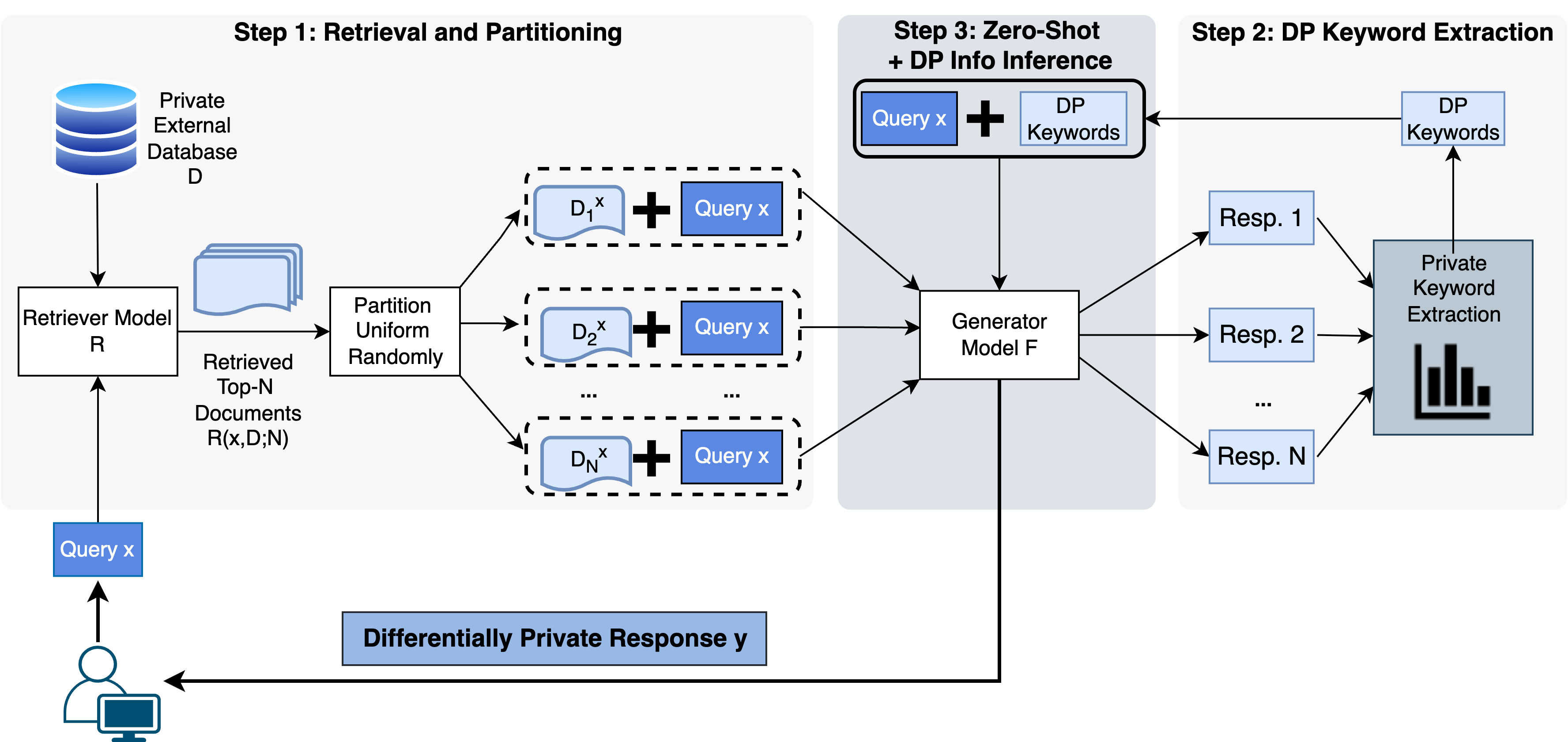}
    \caption{The proposed \method{} framework consists of three steps. First, it retrieves the top-$N$ documents most relevant to the query $\mathbf{x}$ from the private database $D$. Each retrieved document $D^{\mathbf{x}}_i$ is paired with query $\mathbf{x}$ and passed to the generator model $F$ to produce responses. Next, \method{} applies a differentially private mechanism to extract the most frequent keywords from the ensemble of responses. Finally, the selected keywords are combined with query $\mathbf{x}$ and fed back into the generator $F$ to produce the final output $\mathbf{y}$.}
  \label{fig:method_overview}
\end{figure*}

%% file: body/2_prelim.tex
\section{Preliminary}
\label{sec:prelim}

\subsection{Retrieval-Augmented Generation}
\label{prelim:rag}

Retrieval-Augmented Generation (RAG) is a hybrid architecture that augments a large language model (LLM) with an external retrieval mechanism to improve its performance on knowledge-intensive tasks. Given a user query \( \mathbf{x} \), a retriever model \( R \) is first used to fetch the top-\( N \) most relevant documents from a private corpus \( D \). Formally, the retriever returns a set of documents 
\begin{align}
     D_1^{\mathbf{x}}, \ldots, D_N^{\mathbf{x}} \leftarrow R(\mathbf{x}, D; N) 
\end{align}

where each \( D_i^{\mathbf{x}} \in D \) is deemed semantically similar to the query \( \mathbf{x} \). Semantic similarity is typically assessed by computing the distance between the vector representations of a query and a document. Commonly used distance functions include the dot product and cosine similarity.
Each document \( D_i^{\mathbf{x}} \) is then paired with the query and passed into a generator model \( F \), which synthesizes a response:
\begin{align}
    y_i \leftarrow F(D_i^{\mathbf{x}}, \mathbf{x})
\end{align}
The outputs \( y_i \) are natural language responses that reflect the LLM’s understanding conditioned on both the query and the retrieved document.

RAG offers several advantages over standalone language modeling. By retrieving top-\( N \) relevant documents \( D_1^{\mathbf{x}}, \ldots, D_N^{\mathbf{x}} \) from a corpus \( D \) using a retriever \( R \), and conditioning generation on these documents, RAG enables dynamic grounding without retraining the generator \( F \). This allows models to incorporate up-to-date or domain-specific information at inference time, which is especially valuable when \( D \) is frequently updated or too large to encode directly into model parameters.
A key benefit of RAG is its ability to reduce hallucinations. Since outputs \( y_i = F(D_i^{\mathbf{x}}, \mathbf{x}) \) are generated with direct reference to retrieved content, responses are more likely to be factually accurate. 
RAG is also modular as \( R \) and \( F \) can be trained or fine-tuned independently, enabling flexible adaptation across domains. 

\subsection{Differential Privacy}
\label{prelim:dp}
Differential privacy (DP) is the gold standard for reasoning about the privacy of machine learning algorithms. 

\begin{definition}[$(\epsilon, \delta$)-DP \cite{dp}]\label{def:approx_dp}
    For $\epsilon\geq 0, \delta\in[0, 1]$, a randomized algorithm $M$ is $(\epsilon, \delta)$-differentially private if for any pair of adjacent datasets $D$ and $D'$ that differ in only one data point, it holds that $Pr[M(D) \in E] \leq e^{\epsilon} \cdot Pr[M(D') \in E] + \delta$
\end{definition}

The above definition indicates that if two datasets are similar, a DP algorithm should produce similar output $E$ with a high probability so that attackers cannot infer the difference between them. In our case, $M$ functions as a RAG algorithm, producing answers to queries by utilizing private retrieved document as context. If this RAG algorithm adheres to differential privacy, it should generate similar outputs even when the retrieved documents vary. Consequently, this prohibits the generation of private information, such as replicating the retrieved context.


\subsubsection{Private Generation via Sample-and-Aggregate}
\label{dp:sampleandagg}
There has been a body of work on generating token sequences from an LLM with DP, all of which rely on the idea of the sample-and-aggregate framework in DP \citep{sampleandagg}. In this framework, the sensitive dataset is partitioned into pairwise disjoint subsets. Then, the model will perform inference on each subset to generate a response. Lastly, the responses are privately aggregated together using a DP mechanism \cite{duan2023flocks, dpfewshot, dpicl}, such as adding Gaussian noise to the aggregation. The reason for using the sample-and-aggregate framework is that it helps with calculating the global sensitivity, which is an important property needed for DP mechanisms. If we have two neighboring datasets $D$ and $D'$ that differ by one document $d_i$, then at most one subset will differ since the document can only be contained in at most one subset (due to the pairwise disjoint property of sample-and-aggregate). Hence, changing one document changes at most one response, which we can then use to argue the global sensitivity. 

\subsubsection{Private top-$k$ selection}
\label{dp:ptr}
The private top-$k$ selection problem is one of the most fundamental problems in privacy-preserving data analysis. The problem is given a set of candidates with corresponding counts, return the top-$k$ candidates based on the counts in a privacy-preserving way \cite{mcsherry2007mechanism}. In this work, we adopt an adaptive top-$k$ selection algorithm from \citet{zhu2022adaptive}, which is based on the widely-known propose-test-release (PTR) framework \cite{ptr}. The idea is that as long as the count difference between the $k$-th and the ($k+1$)-th highest candidates is larger than 2, we can non-privately release the top-$k$ candidates.



%% file: body/3_problemsettings.tex
\section{Problem Settings}
\label{sec:problemsetting}

In this section, we describe our problem settings and the threat model.

\textbf{Problem Settings. }
Following the problem setting in Figure \ref{fig:problem_setting}, suppose a RAG system receives an input query $\mathbf{x}$ and wants to generate a response using a generator model $F$. Typically $F$ is a decoder-only model such as LLaMA or Qwen. Also suppose the system has access to private external database $D$ to assist with answering the query. The goal is to generate the response $\mathbf{y}$ such that it is $(\epsilon, \delta)$-DP (definition \ref{def:approx_dp}) with respect to a private RAG database $D$.

\textbf{Threat Model. }
Additionally, our work assumes a realistic setup where an adversary can only query the RAG system with any prompt $\mathbf{x}$. The adversary has access only to (1) the response $\mathbf{y}$ and (2) the generator model $F$. Consequently, the adversary does not have access to the external database $D$. We can assume the adversary has access to the generator model, as this is usually publicly available to download from the internet. That is, the pre-training (and fine-tuning) data used to pre-train (and fine-tune) the generator LLM model is considered public and we are only concerned about preserving the privacy of the RAG database $D$. 

%% file: body/4_method.tex
\section{Methodology}
\input{body/method_NQ_histogram}
\input{body/method_TQA_histogram}

\label{sec:method}
In this section, we introduce \method{}, a DP RAG framework. We begin by discussing the motivation behind the design of \method{} in Section \ref{sec:motivation}. Then we will go through \method{} in more detail in Section \ref{method:description}.

\subsection{Motivation}\label{sec:motivation}
\label{method:example}
The challenge of generating differentially private text lies in the high dimensionality of the output space. The generator model $F$ typically generates text autoregressively, i.e. it iteratively samples the next token $y_{t+1}$ from $F$ by using the previously sampled tokens $y_1, ..., y_t$ plus some provided context $c$ to obtain $y_{t+1} \sim F(y_{t+1}|c, y_1, ..., y_t)$. Then we feed $y_{t+1}$ back to the model to sample the next token until the desired stop criteria is met (either max token length reached or eos token sampled). However, the number of possible tokens that $y_{t+1}$ can realize is equal to the vocabulary space of the generator model's tokenizer, which could be upwards to $50,000$. If the max generation length of the response $\mathbf{y}$ is $10$, then the output space of possible responses is $50,000^{10}$. Hence, naively applying DP noise each time we are generating the next token results in a large magnitude of noise, which can have a deleterious effect on the generated tokens and quickly destroy the utility. 

To overcome the large dimensionality of private text generation, the design of \method{} derives from a key observation that most question-answering datasets can be answered sufficiently with just a few keywords. To illustrate this, consider the following question-answering example from the Natural Question (NQ) dataset \citep{nq}, a widely-used RAG benchmarking dataset:

\begin{tcolorbox}[colframe=black, colback=gray!5, boxrule=1pt, arc=4pt]
\textbf{Question:} who lives in the imperial palace in tokyo?

\textbf{Ground Truth Answer:} the Imperial Family
\end{tcolorbox}

From this example, we see that the ground truth only contains three words (i.e., three tokens). Hence, extracting the keywords "Imperial" and "Family" would completely preserve the semantic meaning of the ground truth answer. To demonstrate that this observation holds more generally, Figure \ref{fig:method_nqhist} and Figure \ref{fig:method_tqahist} show the histograms of token lengths of ground truth answers in NQ dataset and Trivia Question-Answering (TQA) dataset \cite{trivia}, respectively. 
As we can see, the histograms are rightly skewed where most of the ground-truth answers contain only one to four tokens. The main takeaway here is that we can convert the task of accurately answering these questions into obtaining only a small set of correct tokens. Such a conversion effectively converts the QA task into a keyword space which can be considered as a low-dimensional approximation of the entire sequence space of correct answers. Consequently, operating in this lower dimensional space will help us preserve the utility as we integrate differential privacy into this space. 

\subsection{\method{}}
\label{method:description}
\input{body/dpkwrag_alg}
\input{body/ptr_alg}

\input{body/top_k_alg}

As shown in Figure \ref{fig:method_overview}, at a high level, our algorithm proceeds in three steps: (1) retrieval and partitioning, (2) DP keyword extraction, and (3) zero-shot with DP information inference. Algorithm \ref{alg:dpkwrag} succinctly describes \method{} and we go through the details below.

\textbf{Retrieval and partitioning.} First, we use a retriever model $R$ to obtain the top $N$ documents from the private database $D$ that are most relevant to the query $\mathbf{x}$ (line \ref{line:ret_docs}). Specifically, each document in the database is encoded into a dense vector representation by the retriever $R$. At query time, the retriever encodes the input query into its own dense vector and ranks documents based on a similarity measure, typically the dot product or cosine similarity, between the query and document vectors. This process identifies the documents most semantically similar to the query. 

Then, for each retrieved document $D^{\mathbf{x}}_i$, we partition them into chunks with each chunk containing a retrieved document and the input query. The document is augmented  with the user query to produce a prompt, which is fed to the generator model $F$ to generate a response (line \ref{line:responses}). Afterwards, we will obtain $N$ responses. 

\textbf{DP keyword extraction.} Next, we extract keywords from the ensemble of responses $\mathbf{y}_i$ $\forall i$ to be used for the final prompt. The intuition behind this approach is that the answers for QA datasets typically contain a few "correct" tokens. Hence, if the retrieved documents are relevant to the query, it is likely that the correct tokens are contained in the ensemble of responses generated based on different disjoint retrieved documents. This is inspired from prior work's use of PTR for DP text generation \cite{dpicl}. 

The key idea is that we convert keyword extraction into a private top-$k$ selection problem, then use existing solutions in this space. To this end, we form a histogram $\mathbf{H}$ by counting the frequency of each word token among the responses based on the individual retrieved documents. Then we use the histogram to obtain and release the top-$k$ tokens with the highest counts in a DP-manner. Specifically, first we adaptively estimate the optimal $\hat{k}$ to use (line \ref{line:adapt_k}), then extract the top-$\hat{k}$ keywords using \texttt{TopKWithPTR} (line \ref{line:join_em}). 

One subtle limitation is that applying a private selection algorithm, such as the exponential mechanism (EM) \cite{mcsherry2007mechanism}, for every keyword will require applying the EM $k$ times, which will decay the privacy loss $\epsilon$. Hence, \texttt{TopKWithPTR} only needs to be applied once to obtain the top-$\hat{k}$ by testing if $\mathbf{H}_{(k)} - \mathbf{H}_{(k+1)} > 2$. If the test holds, then the top-$k$ indices are exactly the same for all the neighboring datasets, so we can release the exact top-$k$ indices without any privatization. However, the test $\mathbf{H}_{(k)} - \mathbf{H}_{(k+1)} > 2$ must be performed in a differentially private way, which is done with PTR (Algorithm \ref{alg:ptr}). If the test fails, then no tokens can be release and hence the final prompt will resort to zero-shot.  

Moreover, the utility of \texttt{TopKWithPTR} is highest when $k$ is chosen to maximize $\mathbf{H}_{(k)} - \mathbf{H}_{(k+1)}$. Hence, $\hat{k}$ is estimated in a data-dependent way by releasing $\text{argmax}_k \mathbf{H}_{(k)} - \mathbf{H}_{(k+1)}$ using EM (Algorithm \ref{alg:rnm-find-k}). Here, $r(k)$ is a regularizer independent of the dataset, e.g., we can set $r(k) = - \infty$ for any $k > 30$ and $k < 15$, if we don't want to return more than 30 or less than 15 tokens. 

\textbf{Zero-shot + DP info inference.} Lastly, we use the DP released top-$\hat{k}$ tokens along with the user query $\mathbf{x}$ to generate the final response from the generator model $\mathbf{y} \gets F(\mathbf{x}, \{w_i\}_{i=1}^{\hat{k}})$ (line \ref{line:final_response}). Note that the final response does not explicitly use the retrieved documents for the final response, only the extracted keywords. If the test in the \texttt{TopKWithPTR} fails, then the final response will not contain any information from the external database.  

%% file: body/method_NQ_histogram.tex
\begin{figure}[!t]
  \centering
    \includegraphics[width=8cm]{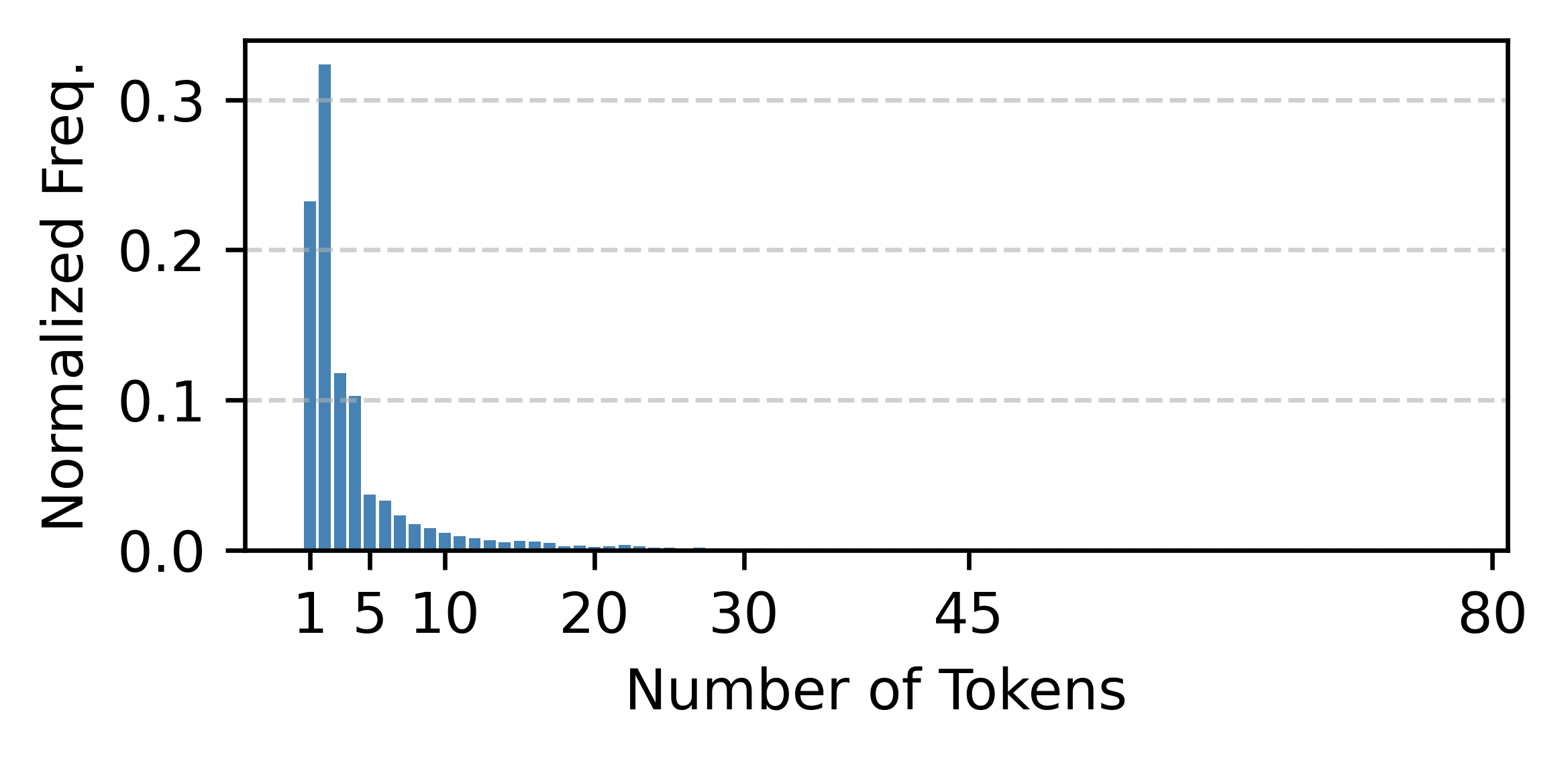}
  \caption{Histogram of token lengths of ground truth answers in NQ dataset.}
  \label{fig:method_nqhist}
\end{figure}

%% file: body/method_TQA_histogram.tex
\begin{figure}[!t]
  \centering
    \includegraphics[width=\linewidth]{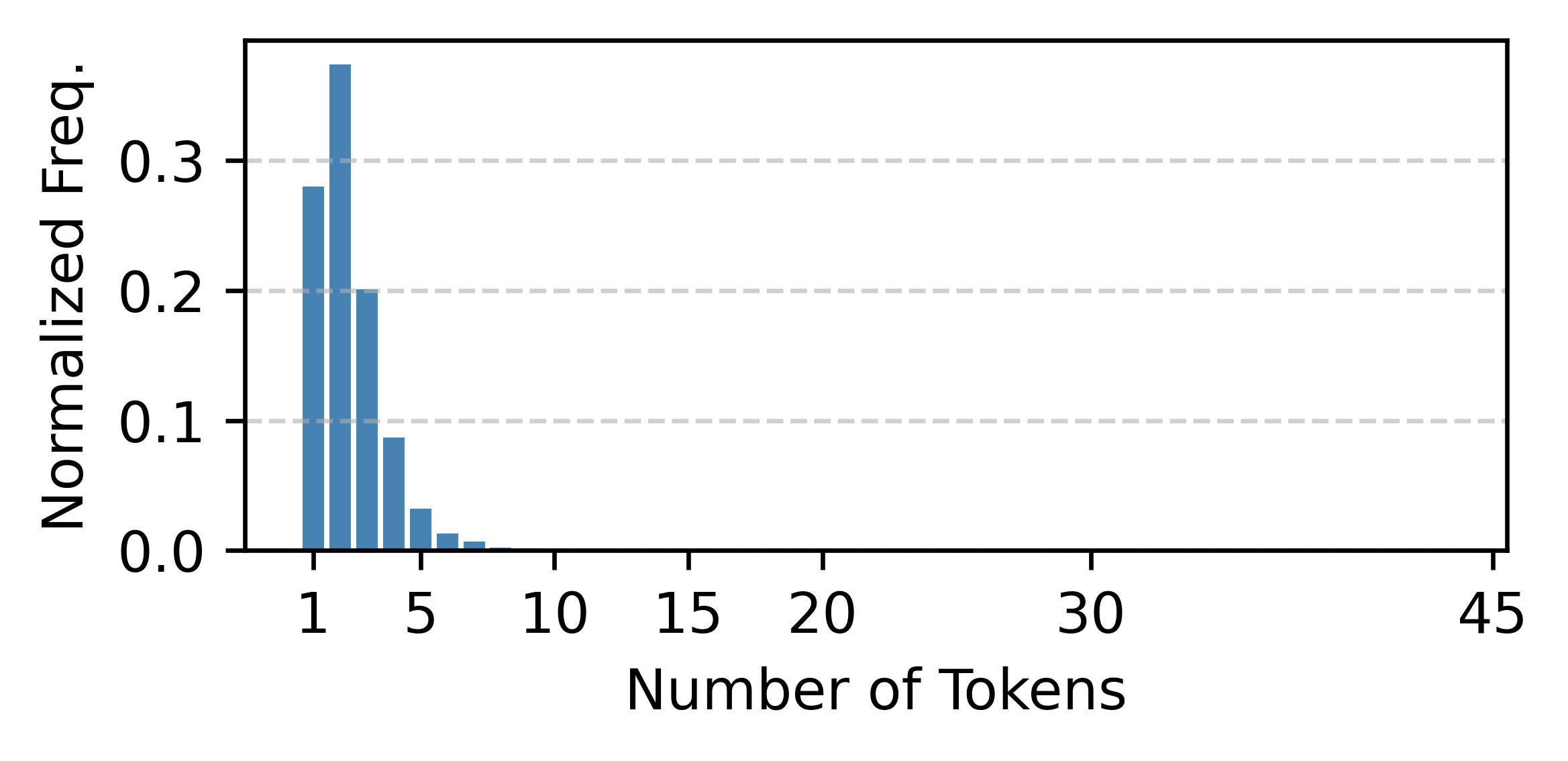}
  \caption{Histogram of token lengths of ground truth answers in TQA dataset.}
  \label{fig:method_tqahist}
\end{figure}

%% file: body/dpkwrag_alg.tex
\begin{algorithm}[!t]
    \caption{\method{} }
    \label{alg:dpkwrag}
    \begin{algorithmic}[1]
        \Require Generator model $F$, retriever model $R$, private database $D$, number of top documents $N$, user query $\mathbf{x}$, max number of generated tokens $T_{\text{max}}$
        \Ensure Generated response $\mathbf{y}$
        \State $D^{\mathbf{x}}_{1}, ..., D^{\mathbf{x}}_{N} \gets R(\mathbf{x}, D; N)$ Retrieve $N$ most relevant documents for query $\mathbf{x}$\label{line:ret_docs}
        \For {$1 \leq i \leq N$}
            \State $\mathbf{y}_i \gets F(D^{\mathbf{x}}_{i}, \mathbf{x})$\label{line:responses}
        \EndFor
        \State Form histogram $\mathbf{H}$ by counting the number of responses $\mathbf{y}_i$ that contain each token. 
        \State $\hat{k} = \texttt{FindBestK}(\mathbf{H})$ \label{line:adapt_k} select optimal number of keywords
        \State $w_{1}, ..., w_{\hat{k}} \gets \texttt{TopKWithPTR}(\hat{k}, \mathbf{H})$ obtain keywords \label{line:join_em}
        \State $\mathbf{y} \gets F(\mathbf{x}, \{w_i\}_{i=1}^{\hat{k}})$ zero-shot with keywords inference
        \State \textbf{Return} $\mathbf{y}$ \label{line:final_response}
    \end{algorithmic}
\end{algorithm}

%% file: body/ptr_alg.tex
\begin{algorithm}[!t]
\caption{\texttt{TopKWithPTR}}
\label{alg:ptr}
    \begin{algorithmic}[1]
      \Require $k$ -- the number of top counted tokens to release; $\mathbf{H}$ -- histogram for the counts of each token; $\delta$ -- failure probability
      \State \textbf{Set} $d_k := \mathbf{H}_{(k)} - \mathbf{H}_{(k+1)}$. 
      \State \textbf{Set} $\widehat d_k := \max(2, d_k) + \mathcal{N}(0, 4\sigma^2) - \Phi(1-\delta; 0, 2\sigma)$. 
    \State \textbf{If} $\widehat d_k > 2$, \textbf{Return} the exact top-$k$ tokens.
    \State \textbf{Else} Terminate with no keywords.
    \end{algorithmic}
\end{algorithm}

%% file: body/top_k_alg.tex
\begin{algorithm}[!t]
    \caption{\texttt{FindBestK}}
    \label{alg:rnm-find-k}
    \begin{algorithmic}[1]
      \Require $\mathbf{H}$ -- histogram for the counts of each token
      \State Compute histogram gap $d_k := \mathbf{H}_{(k)} - \mathbf{H}_{(k+1)}$ for each $k = 1 \ldots N-1$. 
      \State \textbf{Return} $\text{argmax}_k \{ d_k + r(k) + \texttt{Gumbel}(2 /\epsilon) \}$ 
    \end{algorithmic}
\end{algorithm}

%% file: body/5_privacyanalysis.tex
\section{Privacy Analysis}
\label{sec:privacyanalysis}
We now provide a formal guarantee of \method{}, in particular that it achieves $(\epsilon, \delta)$-DP. We provide a proof sketch below and defer the exact details in Appendix \ref{sec:appendix}.

\begin{theorem}
    Let $\mathbf{x}$ be a query received. Suppose we generate a response to $\mathbf{x}$, denoted as $\mathbf{y}$, using \method{} (Algorithm \ref{alg:dpkwrag}) with $F$ as the generator model, $R$ as the retrieval model, and $D$ as the private database. Then \method{} satisfies $(\epsilon, \delta)$-DP with respect to $D$. 
\end{theorem}

\textit{Proof sketch.} The high-level idea of the privacy analysis is that the final response $\mathbf{y}$ is a function of the extracted keywords $w_{1}, ..., w_{\hat{k}}$. The keywords were obtained from the ensemble of responses, which depend on the retrieved documents from the private dataset, and an estimate for the number of keywords $\hat{k}$, which also depends on the responses. Hence, it suffices to show that (1) \texttt{FindBestK} return $\hat{k}$ that is differentially private with respect to $D$, then show (2) \texttt{TopKWithPTR} is differentially private with respect to $D$. In regards to (1), we use the exponential mechanism to achieve differential privacy on $\hat{k}$. And (2) achieves differential privacy using standard analysis of PTR framework. For both (1) and (2), the global sensitivity of the utility function $d_k$ is $2$ because that for two neighboring external datasets $D$ and $D'$ they differ by one document. Subsequently, they will differ by one retrieved document, say the $i$-th $D^{\mathbf{x}}_{i} \neq D'^{\mathbf{x}}_{i}$. Hence the sensitivity of the utility function $d_k$ is $2$. Once we prove these two claims, then we can invoke the post-processing theorem of DP-- stating that a DP quantity does not leak additional privacy about the dataset-- to argue that the final response is DP. 

\begin{remark}
    The formal proof from Appendix \ref{sec:appendix} is relatively straightforward, as it relies on mostly well-established properties and theorems. However, we utilize Renyi Differential Privacy (RDP) \cite{mironov2017renyi} to perform our privacy analysis. Hence, we formally introduce and define these properties in terms of RDP in the Appendix due to spatial constraints. 
\end{remark}

%% file: body/6_experiment.tex
\section{Experiments}\label{sec:exp}

\subsection{Experimental Setup}\label{exp:setup}
\textbf{Dataset. } To evaluate the effectiveness of our proposed method, we conduct experiments on two widely used benchmark datasets in the RAG literature, Natural Questions (NQ) \citep{nq} and Trivia Question Answering (TQA) \citep{trivia}. Both datasets are designed for open-domain question answering and consist of a diverse set of real-world questions, each associated with multiple reference answers. Following standard RAG evaluations \citep{prerag,rag,ragged}, we use the Wikipedia corpus as our external knowledge source for retrieval. Due to practical constraints in computational resources, we follow prior work \citep{dprag1} and use a subset of 100 questions from each dataset. Note that we filter out questions with empty ground-truth references in both datasets in order to correctly compute the evaluation metrics. 


\textbf{Model Architectures. } Our system follows a standard retrieval-augmented generation (RAG) pipeline, composed of a retriever and a generator. For the retriever component, we adopt the Dense Passage Retriever (DPR) \citep{dpr}, a widely used dual-encoder architecture built on top of BERT \citep{bert}. DPR encodes both the input question and the candidate passages into dense vector embeddings and retrieves the most relevant documents by measuring similarity in the embedding space. In our experiment, we use the inner product as the similarity metric. In addition, given the large size of the Wikipedia corpus (21 million passages), we use FAISS \cite{faiss} to accelerate retrieval. FAISS leverages approximate nearest neighbor algorithms for efficient similarity search in high-dimensional embedding spaces.

For the generator, we compare several state-of-the-art large language models that have been instruction-tuned to follow user prompts. Specifically, we evaluate Qwen 2.5 (3B) \citep{qwen25}, Llama 3.2 (3B), and Llama 3.1 (8B) \citep{llama3}. The instruction-tuned versions of these models are particularly appropriate for QA tasks, as they are optimized to respond effectively to task-specific prompts.

\input{body/main_result_NQ}

\textbf{Baselines. } We evaluate the effectiveness of \method{} (Algorithm \ref{alg:dpkwrag}) by comparing it against three representative baselines, each reflecting a different level of privacy and retrieval capability.

The first baseline, denoted as \textbf{non-RAG} $\boldsymbol{(\epsilon = 0)}$, represents a strictly private setup where only the question is provided to the language model, without any access to external documents. Since no retrieval is involved, this configuration guarantees inherent privacy and serves as a minimal baseline. For \method{} to be considered practically useful, it must achieve better performance than this no-retrieval baseline under reasonable privacy budgets.

The second baseline, labeled as \textbf{KSA} $\boldsymbol{(\epsilon = \infty)}$, performs the same keyword extraction as \method{} but releases the top-$K$ keywords without introducing any noise. This non-private variant effectively removes the privacy constraints while preserving the aggregation structure of our method, making it a strong upper bound for our privacy-preserving approach.

The third baseline is the standard retrieval-augmented generation pipeline, denoted as \textbf{RAG} $\boldsymbol{(\epsilon = \infty)}$. In this setting, the question is paired with the top-2 retrieved documents from the external knowledge base, and the full documents are fed directly into the language model without any privacy protection. This setting represents a theoretical upper bound in terms of utility, as it leverages full retrieval with no noise or low-dimensional approximation.


\textbf{Settings. } We follow previous work \citep{dprag1} and set $\delta = 10^{-4}$. We select $\epsilon = \{1,2,3,5,8\}$ to achieve different levels of privacy. We set the number of ensembles to 80 for our method \method{} and non-private KSA method, which is found to be the optimal number of ensembles and detailed later in the ablation studies in section \ref{ablation:ensemble} . 

\textbf{Metrics. } Following prior works \citep{dpicl,ragf1}, we adopt four widely used evaluation metrics, F1, ROUGE-1, ROUGE-L and normalized Levenshitein similarity. These metrics collectively capture both lexical overlap and semantic alignment between generated outputs and ground-truth references. 
The F1 score computes the harmonic mean of precision and recall, and is particularly well-suited for QA tasks where partial correctness is meaningful. It reflects how well the predicted answer tokens align with those in the reference, rewarding both completeness and accuracy. 
ROUGE-1 measures unigram (i.e., single word) overlap between the prediction and the reference text. This provides a straightforward measure of lexical similarity. ROUGE-L, on the other hand, incorporates the sequential nature of language by computing the longest common subsequence (LCS), thus capturing not only which words are shared, but also whether they appear in a similar order, which is an important indicator of fluency and coherence in natural language generation. 
Levenshtein similarity is derived from the Levenshtein distance, which quantifies the number of single-character insertions, deletions, or substitutions required to transform one string into another. By normalizing this value, we obtain a score that reflects how closely the predicted answer resembles the reference at the character level, offering a fine-grained perspective on textual similarity. Across all four metrics, higher scores indicate better alignment between the generated and ground-truth answers, and thus better performance.

\input{body/main_result_TQA}
\subsection{Privacy-Utility Tradeoff of \method{}}\label{exp:results}
We show the privacy-utility tradeoff of \method{} across four evaluation metrics and three different LLMs on NQ and TQA datasets. We evaluate \method{} with different privacy budget $\epsilon$, ranging from 1 to 8. Smaller $\epsilon$ indicates stronger privacy constraints. We also compare our method with the non-RAG ($\epsilon=0$), non-private KSA ($\epsilon=\infty$) and non-private RAG ($\epsilon=\infty$) baselines.

\subsubsection{Performance on NQ dataset. }\label{results:nq}
Figure \ref{fig:exp_nq} presents the performance of \method{} on the NQ dataset. Across all models and metrics, we observe a consistent performance improvement as the privacy budget $\epsilon$ increases from 1 to 8, confirming that \method{} effectively trades off privacy for utility. Notably, at moderate privacy levels (e.g. $\epsilon =$ 3 or 5), \method{} substantially outperforms the strictly private non-RAG baseline, demonstrating its ability to retain useful keywords even under noise. For instance, with Llama 3.2 (3B) as the generator, the F1 score improves from 21.67 at $\epsilon = 0$ to 22.55 at $\epsilon = 2$, and further to 25.18 at $\epsilon = 8$, highlighting a strong privacy-utility tradeoff. However, at $\epsilon = 1$, performance tends to drop below or remain comparable to the non-RAG baseline ($\epsilon = 0$). This is likely due to the higher noise level required under strong privacy constraints, which hampers the success of the propose-test-release (PTR) condition $\mathbf{H}_{(k)} - \mathbf{H}_{(k+1)} > 2$, resulting in fewer or no keywords being released to support final response generation.

We find that \method{} often matches or even surpasses the performance of the non-private KSA ($\epsilon = \infty$) baseline, despite the presence of differential privacy noise. For example, using Qwen 2.5 (3B) as the generator, \method{} achieves ROUGE-1 and ROUGE-L scores at a moderate privacy budget ($\epsilon = 5$) that are comparable to or exceed those of non-private KSA. This suggests that a noisy histogram, constructed across ensembles of multiple responses, effectively preserves key semantic contents even in the presence of differential privacy noise. Additionally, unlike non-private KSA which deterministically selects the most frequent keywords, \method{}'s stochastic keyword release mechanism introduces variation across runs, which may help reduce overfitting to overly dominant but potentially less informative tokens. Overall, these findings underscore that controlled randomness applied over an ensemble of semantically coherent outputs enables strong utility-privacy tradeoffs. 

It is also worth noting that the non-private KSA baseline generally falls short of the upper bound set by the non-private RAG ($\epsilon=\infty$) baseline. This performance gap is likely stems from the information loss inherent in compressing an ensemble of responses into a fixed set of keywords. Exploring more effective strategies for utilizing these extracted keywords to further narrow the performance gap remains as an important direction for future work.

Interestingly, we observe that stronger generator models benefit more from \method{}. As we move from Qwen 2.5 (3B) to Llama 3.1 (8B), not only do we observe higher absolute scores, but also greater robustness to DP noise. This supports the intuition that larger models or certain model families are better equipped to infer context and generate meaningful outputs from partially informative or compressed prompts, making them especially well-suited for privacy-preserving QA systems.

Finally, although all four evaluation metrics follow consistent upward trends as $\epsilon$ increases, the degree of improvement varies. Levenshtein similarity, which captures fine-grained character-level differences, improves more noticeably, suggesting that the generated responses become semantically closer to the references as privacy budget increases. ROUGE-L and F1 also show steady improvements, confirming the \method{} yields gains across both string similarity and task-specific QA measures.

\subsubsection{Performance on TQA dataset. }\label{results:tqa}
Figure \ref{fig:exp_tqa} shows the performance of \method{} on the TQA dataset. Similar to the trends observed on NQ, we see consistent performance gains as the privacy parameter $\epsilon$ increases from 1 to 8, illustrating \method{}'s ability to effectively trade privacy for utility. Beginning at a privacy budget of $\epsilon = 2$, \method{} consistently outperforms the strictly private non-RAG ($\epsilon = 0$) baseline, demonstrating that even under strong privacy constraints, meaningful keywords can still be preserved. One example is when using Llama 3.2 (3B) as the generator, the F1 score improves from 55.12 at $\epsilon = 0$ to 56.92 at $\epsilon = 5$, indicating a strong privacy-utility tradeoff. Notably, unlike on NQ where performance at lower $\epsilon$ values occasionally regresses, \method{} on TQA exhibits smoother and more stable improvements even under tight privacy budgets. 

A particular interesting finding is that \method{} performs competitively against the non-private RAG ($\epsilon = \infty$) baseline in some cases. For instance, with Qwen 2.5 (3B), the ROUGE-1, ROUGE-L, and Levenshtein similarity scores at $\epsilon = 5$ exceed those of the non-private RAG baseline, despite the added differential privacy noise. This suggests that using keywords from a noisy histogram formed across an ensemble of responses can yield final responses that are semantically rich as those generated with full document contexts, particularly when the retrieved context is noisy or contains more redundancy.

Digging deeper into the baselines, we observe that the non-private KSA baseline actually outperforms non-private RAG in smaller models such as Qwen 2.5 (3B) and Llama 3.2 (3B), while the trend reverses for larger models like Llama 3.1 (8B). This difference can be attributed to the relative contributions of two information sources, namely the consensus of keywords captured by keyword frequency across the ensemble of 80 responses, and the parametric knowledge embedded in the generator LLM. For smaller models, the ensemble-driven signal dominates, giving KSA an advantage. However, as the generator becomes more powerful, its internal knowledge plays a larger role in improving generation quality and thus shifting the advantage toward full-context RAG.

In line with these observations, \method{} shows greater benefits with stronger generator models. Moving from Qwen 2.5 (3B) to Llama 3.1 (8B), we see not only higher absolute performance but also improved stability across different privacy budgets. This reinforces our earlier conclusion from the NQ results that larger and more capable models are better suited to generalize from sparse or compressed prompts, making them strong candidates for privacy-preserving RAG.

Lastly, all four metrics display consistent improvement as $\epsilon$ increases, but the extent of gains varies. Levenshtein similarity and ROUGE-L tend to improve more steeply, indicating that as more keywords are released, the generated final responses become semantically and structurally closer to the reference answers. F1 and ROUGE-1 also improve steadily, confirming that \method{} remains effective across both QA-style metrics and surface-level lexical overlap. Compared to NQ, \method{} on TQA dataset achieves significantly higher absolute scores, likely due to its more extractive answer structure and the higher redundancy in supporting context, which together make it easier for \method{} to preserve relevant information even under privacy constraints.

\input{body/main_result_PTR}
\subsubsection{PTR test pass rate.}\label{result:ptr_test}
We analyze the effect of various privacy constraints of \method{} on PTR test pass rates on both NQ and TQA datasets. Following the description of our method \method{} in section \ref{method:description}, the propose-test-release (PTR) test checks if the condition $\mathbf{H}_{(k)} - \mathbf{H}_{(k+1)} > 2$ holds. Figure \ref{fig:exp_ptr} shows the pass rate changes for different privacy parameters $\epsilon$, ranging from 1 to 8, across three different generator LLMs. Across both datasets, we observe a clear and consistent trend that the PTR test pass rate increases steadily as $\epsilon$ increases. This indicates that relaxing the privacy constraint reduces the amount of DP noise added, enabling more privately selected keywords to pass the PTR test and being released to aid the final model outputs. In other words, lower noise leads to more stable test outcomes and thus higher pass rates. This observation also supports our discussion earlier in the privacy-utility tradeoff performance on both datasets.

When comparing across models, we find that Qwen 2.5 (3B) trails slightly behind both Llama models on NQ and TQA. The Llama 3.2 (3B) model generally achieves the highest pass rates on the NQ dataset, especially showing steep gains between $\epsilon=1$ and $\epsilon=5$. In contrast, on the TQA dataset, the Llama 3.1 (8B) model slightly outperforms the others at higher privacy budgets ($\epsilon=5$ and $8$), suggesting that its larger capacity helps identify consistent keyword patterns even under significant DP noise. These trends highlight that stronger models, particularly those in the Llama family, exhibit better resilience to DP noise and are more effective at preserving semantic signals necessary for keyword extraction.

In terms of dataset differences, we observe that the TQA dataset consistently achieves higher PTR pass rates than NQ across all models and $\epsilon$ values. Even at $epsilon=1$, models on TQA start from a noticeably higher baseline. This discrepancy likely stems from differences in dataset structure. TQA answers tend to be more extractive and contextually redundant, which helps generate more consistent keyword distributions across ensembles. As a result, keywords in TQA are easier to identify under the added DP noise.

\subsection{Ablation Studies}\label{exp:ablation}
We also conduct ablation studies with varying model sizes and number of ensembles on NQ dataset to investigate how these hyper-parameters affect the performance of our proposed method \method{}. 

\subsubsection{Effect of Model Size}\label{ablation:modelsize}
Figure \ref{fig:ablation_model} illustrates the impact of model size on the performance of \method{} across all four metrics, F1, ROUGE-1, ROUGE-L and Levenshtein similarity, under a fixed privacy budget of $\epsilon=3$, along with non-private baselines for comparison. In this experiment, we evaluate the Qwen 2.5 model family at four scales, 1.5B , 3B, 7B and 32B. To accommodate hardware memory limitations, we use 4-bit quantization to the 32B model weights, ensuring that all models can be evaluated under the same computational constraints. 

As expected, larger models consistently achieve better performance across most metrics. As model size increases, we observe steady gains in F1 and ROUGE scores under the privacy-preserving setting, indicating that more expressive LLMs are more capable of compensating for the restricted input caused by providing only keywords as context in the prompts. Notably, even under a moderate privacy constraint of $\epsilon=3$, the 32B model nearly matches the performance of the fully non-private RAG baseline ($\epsilon=\infty$), demonstrating that \method{} benefits substantially from scaling to larger model capacities.

Interestingly, however, we observe a decline in Levenshtein similarity as model size increases to 32B. This seemingly counterintuitive trend can be explained by the generative behavior of larger models. Bigger LLMs, such as Qwen 2.5 (32B), are more likely to paraphrase reference answers with greater fluency and lexical diversity. While these outputs may be semantically accurate, they diverge more from the reference at the character level, which directly lowers Levenshtein similarity. Larger models may also produce longer or more elaborative responses, introducing more edits even if semantically valid, while smaller models are more conservative and literal in their generations. Therefore, the drop in Levenshtein similarity should be interpreted not as a sign of degraded quality, but rather as a side effect of more fluent and paraphrastic generation.

At the other end of the scale, the smallest model size (1.5B) reveals an unusual pattern. The non-RAG baseline ($\epsilon=0$), where only the question is given without any retrieved context, outperforms both the non-private KSA and \method{} ($\epsilon=3$) settings. This can be explained by the weaker instruction-following capabilities of smaller models. When provided with prompts using keyword-only context, whether deterministic (KSA) or noisy (\method{}), these models struggle to interpret and integrate sparse cues into coherent answers. In contrast, the non-RAG setting, which presents only the raw question, may align better with the pretraining distribution of small models and avoid potential confusion from prompts with incomplete or overly narrow keywords as context. Consequently, the model relies more heavily on its parametric knowledge, sometimes yielding better outputs than poorly constructed context.

\input{body/ablation_model_size}
Moreover, between KSA and \method{} ($\epsilon=3$) at 1.5B, the latter surprisingly performs better. One possible reason is that the stochastic nature of \method{} introduces diversity into the prompts, occasionally surfacing less frequent but semantically useful tokens. This noise-driven variation can act as a form of regularization that prevents the model from overfitting to dominant but uninformative keywords and this is a vulnerability particularly pronounced in smaller models.

Overall, these findings underscores that larger models are not only more effective at leveraging prompts with sparse context but also more robust to the limitations imposed by differential privacy. They suggest that \method{} is best paired with instruction-following-capable LLMs to achieve strong privacy-utility tradeoffs.

\input{body/ablation_ensemble}
\subsubsection{Effect of Number of Ensembles}\label{ablation:ensemble}

Figure~\ref{fig:ablation_ensemble} shows how the number of retrieval ensembles impacts the performance of \method{}. In this experiment, we use the Llama 3.2 (3B) generator model and fix the privacy budget at $\epsilon = 3$, while varying the number of ensembles from 10 to 100. The number of ensembles reflects how many top-ranked documents are used from retrieval. For instance, an ensemble size of 10 uses the top 10 retrieved documents, while a size of 50 incorporates the top 50. Each ensemble corresponds to a different retrieved document used to prompt the generator, and their aggregated responses are used to construct a noisy token histogram for differentially private keyword selection.

We observe a clear progression in performance trends across all four evaluation metrics as the number of ensembles increases.

When using between 10 and 40 ensembles, performance remains relatively stable. This indicates that aggregating a limited number of ensembles does not provide sufficient signal strength to overcome the DP noise introduced. With few ensembles, informative tokens may appear sporadically across responses, leading to sparse or inconsistent histograms where the DP mechanism struggles to identify high-frequency keywords reliably.

As the number of ensembles increases from 40 to 80, performance improves steadily. The addition of more ensembles enriches the token histogram, making it more representative and resilient to DP noise. Recurrent tokens consistently appearing across semantically relevant responses emerge more clearly, enabling the stochastic keyword release mechanism to retain more meaningful and informative content. The growing diversity of responses within this range plays a key role in stabilizing the keyword selection process under DP constraints.

Between 80 and 100 ensembles, performance begins to plateau. While more ensembles continue to be aggregated, the marginal gains diminish. Additional documents may contribute less relevant or redundant information, increasing the presence of low-utility or noisy tokens in the histogram. This can dilute the quality of the extracted keywords and limit further improvements in the final output quality.

These findings indicate that while increasing the number of ensembles can significantly enhance performance under DP, there are diminishing returns beyond a certain point. Once the ensemble size becomes large enough to establish a stable token distribution, adding more documents does little to improve and may even slightly degrade the quality of the extracted keywords due to the introduction of irrelevant or noisy content. This highlights the need for more strategic methods for selecting the number of ensembles to further boost performance without compromising the privacy guarantee.

%% file: body/main_result_NQ.tex
\begin{figure*}[!ht]
  \centering
    \includegraphics[width=\linewidth]{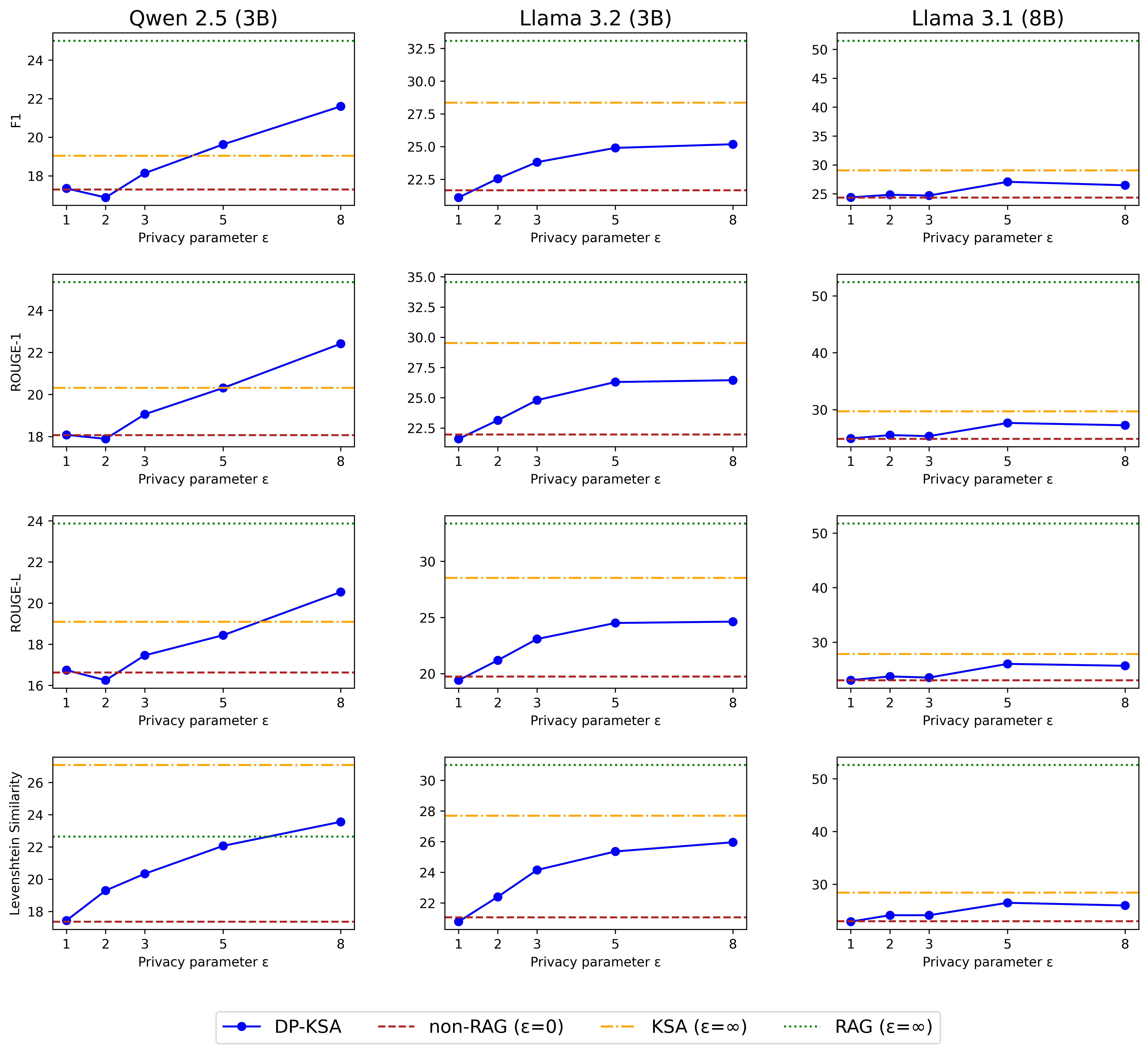}
  \caption{Results of \method{} on NQ dataset with different generator LLMs: Qwen 2.5 (3B), Llama 3.2 (3B), and Llama 3.1 (8B). We use three baselines including non-RAG ($\epsilon=0$), non-private RAG with top-2 retrieved documents ($\epsilon=\infty$), and non-private KSA ($\epsilon=\infty$).}
  \label{fig:exp_nq}
\end{figure*}

%% file: body/main_result_TQA.tex
\begin{figure*}[!ht]
  \centering
    \includegraphics[width=\linewidth]{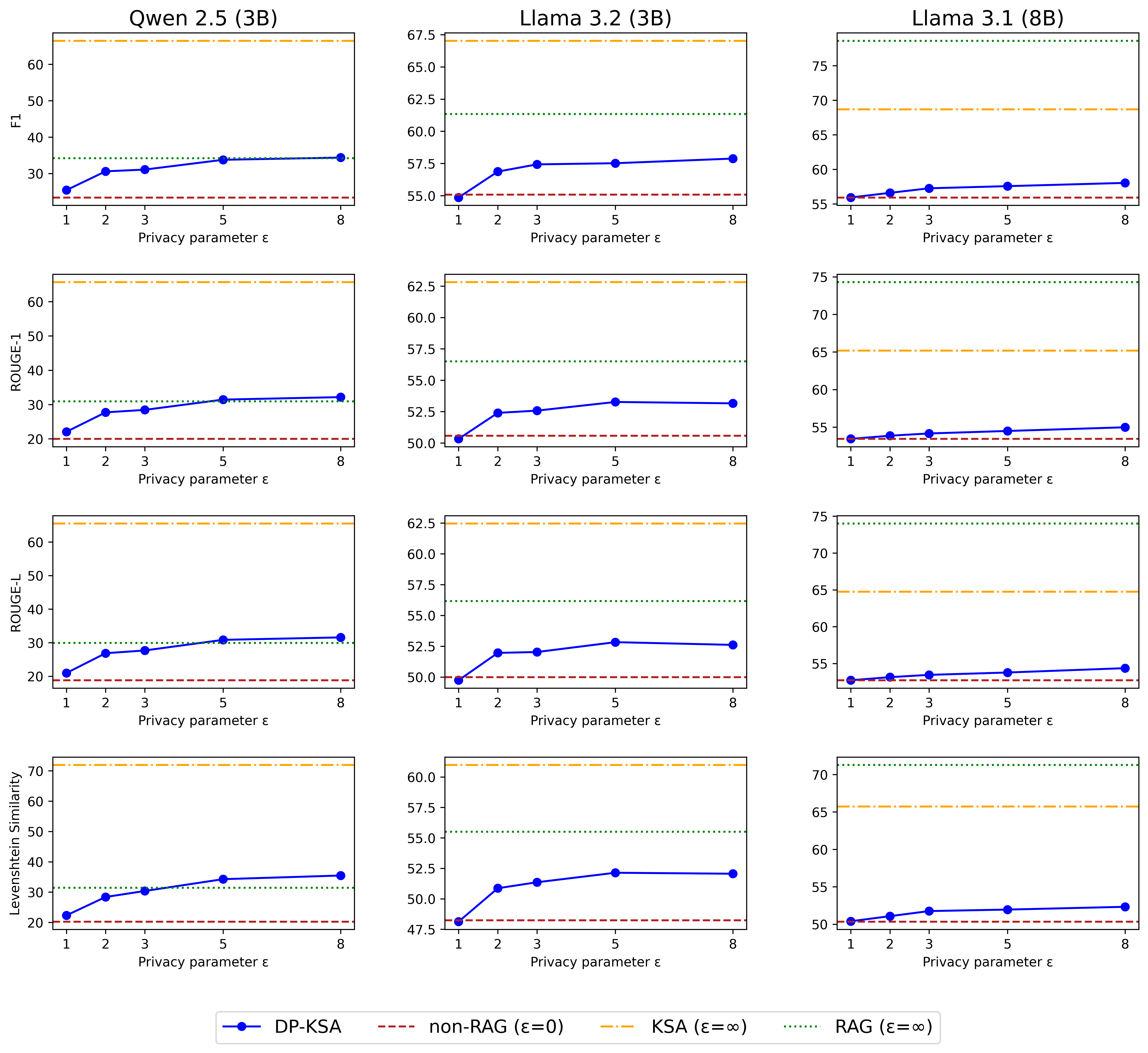}
  \caption{Results of \method{} on TQA  dataset with different generator LLMs: Qwen 2.5 (3B), Llama 3.2 (3B), and Llama 3.1 (8B). We use three baselines including non-RAG ($\epsilon=0$), non-private RAG with top-2 retrieved documents ($\epsilon=\infty$), and non-private KSA ($\epsilon=\infty$).}
  \label{fig:exp_tqa}
\end{figure*}

%% file: body/main_result_PTR.tex

\begin{figure}[!h]
  \centering
  \begin{subfigure}{.45\textwidth}
    \centering
    \includegraphics[width=.9\linewidth]{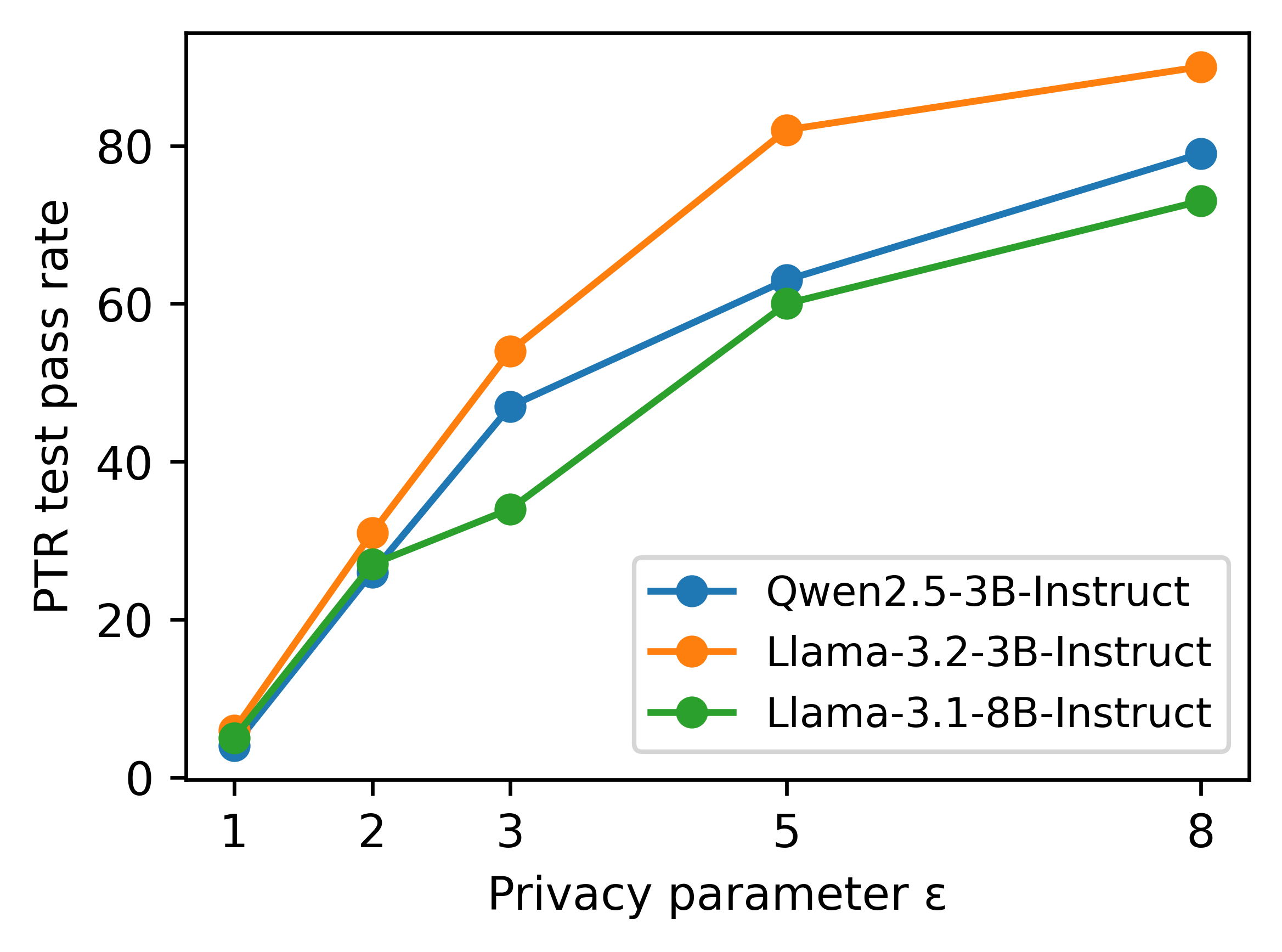}
    \caption{NQ}
  \end{subfigure}
  \hfill
  \begin{subfigure}{.45\textwidth}
    \centering
    \includegraphics[width=.9\linewidth]{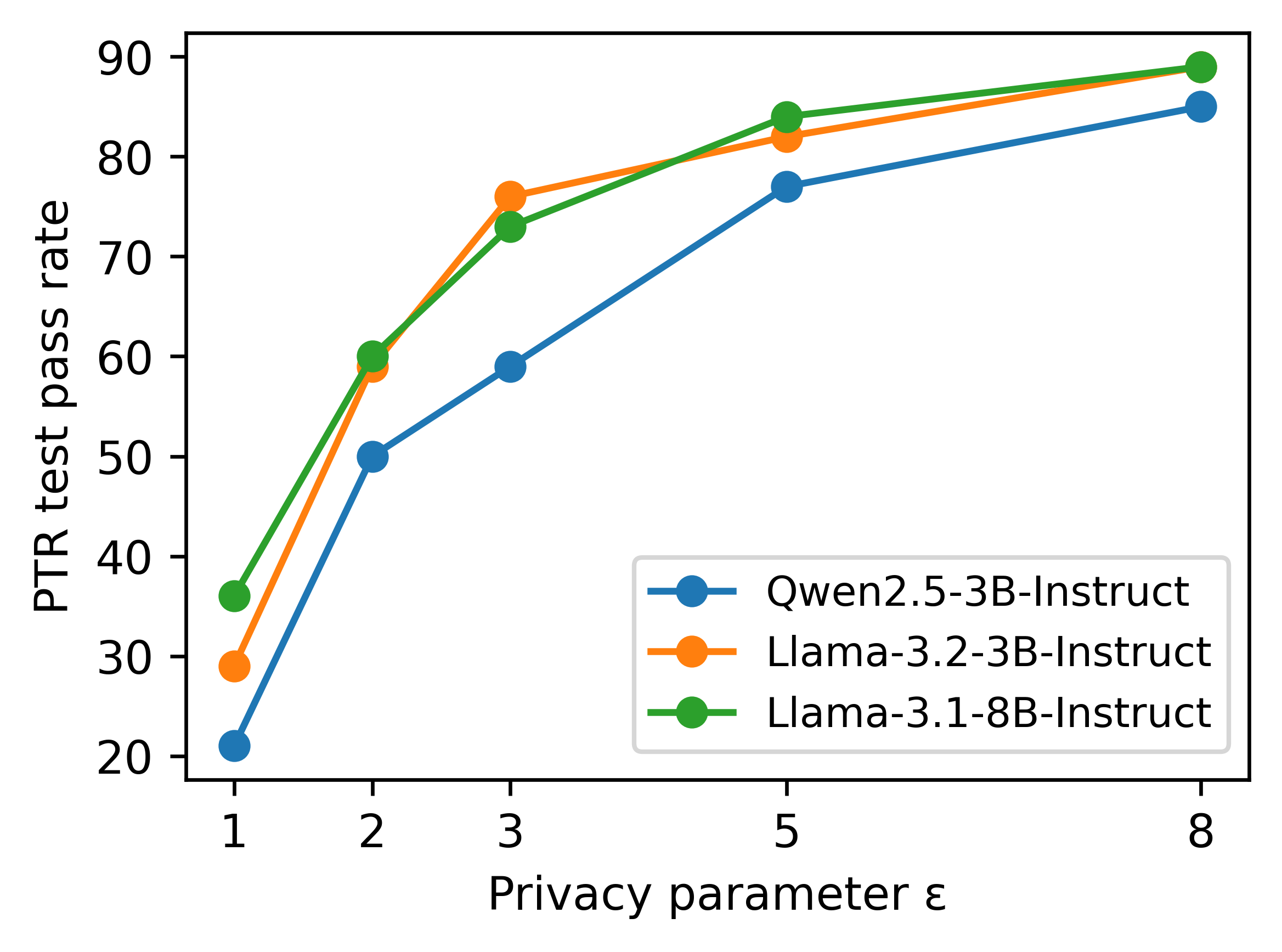}
    \caption{TQA}
  \end{subfigure}
  \caption{Propose-test-release (PTR) test pass rate of \method{} with varying privacy parameters $\epsilon$ on NQ and TQA datasets. We report the results with three different LLMs:  Qwen 2.5 (3B), Llama 3.2 (3B), and Llama 3.1 (8B).}
  \label{fig:exp_ptr}
\end{figure}


%% file: body/ablation_model_size.tex

\begin{figure}[!t]
  \centering
  \includegraphics[width=\linewidth]{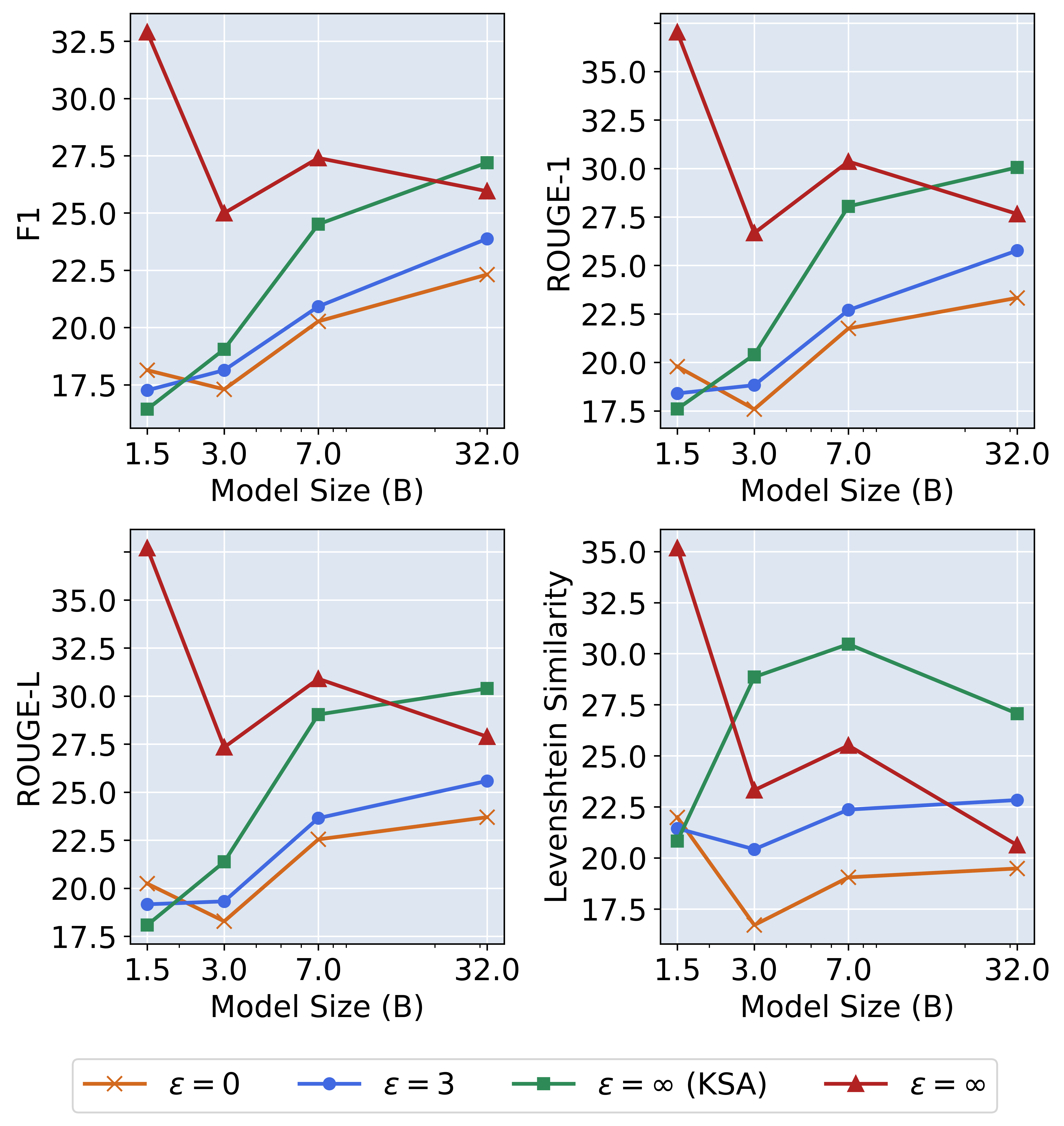}
  \Description{Plot showing ablation studies on model sizes with Qwen 2.5 model family on NQ dataset. Model size axis is in log scale.}
  \caption{Ablation studies on model sizes with Qwen 2.5 model family on NQ dataset. Model size axis is in log scale.}
  \label{fig:ablation_model}
\end{figure}

%% file: body/ablation_ensemble.tex
\begin{figure}[!h]
  \centering
    \includegraphics[width=0.9\linewidth]{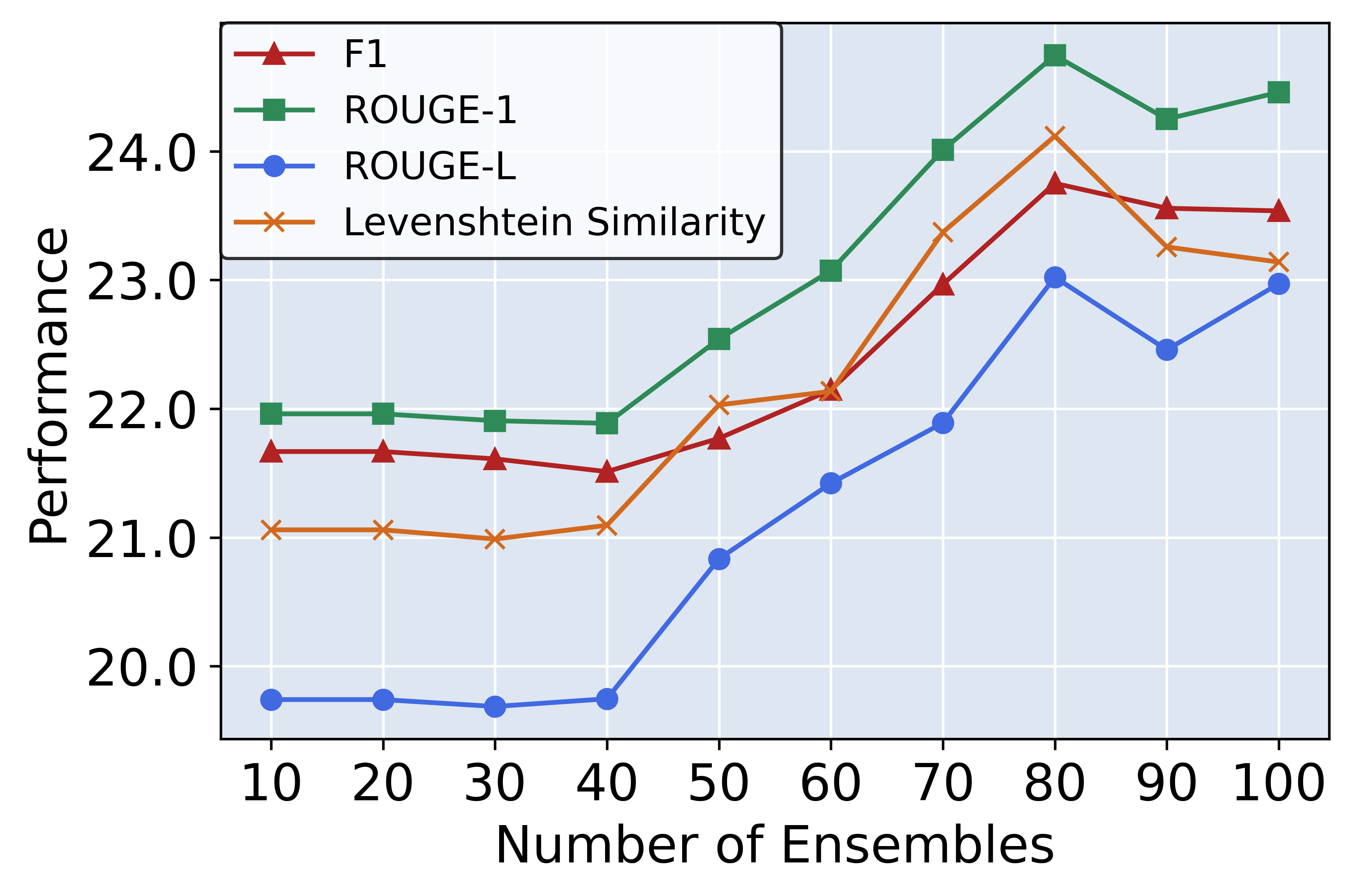}
  \caption{Ablation studies on number of ensembles with Llama 3.2 (3B) model and $\epsilon = 3$ on NQ dataset.}
  \label{fig:ablation_ensemble}
\end{figure}

%% file: body/7_relatedwork.tex
\section{Related Work}
\label{sec:relatedwork}

\subsection{Privacy Attacks on Large Language Models}\label{relatedwork:llmattack}

Prior work has shown that large language models (LLMs) are vulnerable to privacy breaches through a variety of attack vectors. \citet{llmattack} demonstrated that adversaries can extract memorized training examples, such as email addresses, phone numbers, and other sensitive content, by carefully crafting input prompts. These training data extraction attacks reveal that even models trained on supposedly de-identified datasets can still memorize and regurgitate specific private data points.
Prompt extraction attacks \citep{promptattack1, promptattack2, promptattack3} have also received significant attention. These attacks target deployed systems where users provide prompts that encode proprietary logic, credentials, or task-specific templates. Adversaries can use model inversion, gradient leakage, or black-box probing to infer the original prompts, threatening the confidentiality of model customization and user inputs in commercial APIs.

In the context of retrieval-augmented generation (RAG), recent work has highlighted privacy risks stemming from both the retrieval and generation components. \citet{ragattack3} analyzed $k$NN-LMs and showed that retrieved neighbors from the training corpus can expose private information, especially when the distance function reveals document structure or content characteristics. \citet{ragattack2} conducted a systematic investigation into open-source RAG pipelines and found that sensitive data in the retrieval corpus can propagate into generated responses without explicit prompts, particularly when the retriever ranks memorized or high-sensitivity content highly. \citet{ragattack1} extended these findings to commercial production-level RAG systems, showing that the configuration of RAG, including retrieval granularity, document ranking, and prompt construction, can significantly influence the model's susceptibility to data leakage.
Together, these studies underline the importance of privacy-aware design to mitigate the growing threat of information leakage in retrieval-augmented language generation.

\subsection{Privacy-Preserving Large Language Models}\label{relatedwork:ppllm}
Differential privacy has been studied in many tasks in large language models. The differentially private pretraining and finetuning of LLMs have been studied to address the privacy concern in the training data by deploying DP-SGD \citep{dpsgd}. In this paradigm, noise is introduced to the gradient during the model’s training to ensure privacy. However, as the scale of the large language models significantly increased, memory becomes a large bottleneck and makes this approach more challenging in practice. Although recent methods have been proposed for efficient per-example gradient clipping \citep{dpsgdllm} and parameter-efficient fine-tuning \citep{dpsgdft}, it remains a topic of ongoing research in order to address the engineering and optimization problems introduced by DP-SGD. In-context learning adapts to different tasks by illustrating some examples in the context as the task description. DP in-context learning considers the situation when the examples are picked from any private dataset. \citep{dpfewshot} tackles this problem by generating synthetic examples with DP. \citep{dpicl} instead uses a sample-and-aggregate algorithm to generate DP responses.

To mitigate privacy risks in RAG systems, \citep{syntheticrag} proposed an empirical privacy-preserving algorithm for RAG through the synthetic data generation, while our work studies privacy-preserving RAG in the framework of differential privacy, which protects the privacy of each individual document with the theoretical guarantee. Closely-related prior works have proposed DP solutions in the RAG setting by applying DP at every token generation of the LLM \citep{dprag1, dprag2} after retrieving relevant documents. While they offer privacy protection for RAG systems, they come with noticeable limitations. Per-iteration private voting requires composing the privacy loss over the number of tokens being generated, which quickly destroys the privacy-utility tradeoff.

%% file: body/8_conclusion.tex
\section{Conclusion}

In this paper, we presented \method{}, a novel privacy-preserving algorithm that ensures differential privacy for sensitive external data source in the RAG system, enabling us to enhance LLMs by domain-specific but sensitive external data source. \method{} privatizes RAG by extracting the most frequent keywords based on the "propose-test-release" paradigm and augmenting them into the prompt to generate the final output. The privately extracted keywords effectively compresses the semantic
space while still retaining key information pertaining to the retrieved contexts for better utility. The experiments on QA benchmarking datasets show that our algorithm outperforms the non-RAG method under moderate privacy budgets across different models, demonstrating its effectiveness for maintaining high generation quality while providing formal privacy guarantees. We also explored and evaluated the impact of generator model sizes and number of ensembles on our method.

In future work, we will consider a more adaptive scheme for selecting the number of ensembles to further improve privacy-utility tradeoffs across different generator model families and sizes. While \method{} targets scenarios where the external retrieval corpus is private and the generator LLM is trained on publicly available data, commonly seen in real-world RAG systems using open-source models, it is also worth investigating how \method{} affects the training data leakage risks when the generator LLM is pre-trained or fine-tuned on private datasets.

%% file: body/9_appendix.tex
\clearpage 
\appendix

\section{Privacy Analysis of \method{}}
\label{sec:appendix}
We now provide technical details behind the privacy analysis of \method{}. We first introduce some definitions and theorems to help with the analysis.

\subsection{Relevant Properties}

In the main text, we introduced $(\epsilon, \delta)$-DP (Definition \ref{def:approx_dp}). Additionally, we will introduce Rényi Differential Privacy (RDP), a variant of $(\epsilon, \delta)$-DP that uses Rényi divergence to measure the difference between $M(D)$ and $M(D')$. 

\begin{definition}[Rényi Divergence]
   For two probability distributions $P$ and $Q$ defined over $\mathcal{R}$, the Rényi divergence of order $\alpha > 1$ is
   \begin{equation*}
       D_{\alpha}(P || Q) = \frac{1}{\alpha -1}\log \mathbb{E}_{x \sim Q} \left [\left(\frac{P(x)}{Q(x)} \right)^{\alpha}\right ].
   \end{equation*}
\end{definition}

\begin{definition}($(\alpha, \epsilon(\alpha))$-RDP \cite{mironov2017renyi})
    A randomized algorithm $M: \mathcal{D} \rightarrow \mathcal{R}$ is $(\epsilon(\alpha), \alpha)$-RDP if for any adjacent datasets $D, D'\in \mathcal{D}$ it holds that $D_{\alpha}(M(D)||M(D'))\leq \epsilon(\alpha).$
\end{definition}

An advantage of RDP is its convenient composition properties, which states that the privacy loss of the composition of multiple RDP algorithms is simply the sum of each RDP algorithm. We state this more formally below.

\begin{theorem}[Composition \cite{mironov2017renyi}]\label{thm:composition}
    Let $A_1, ..., A_k$ be a sequence of $(\alpha, \epsilon(\alpha))$-RDP algorithms. Then the composition $A_k \circ A_{k-1} \circ ... \circ A_1$ is $(\alpha, k\epsilon(\alpha))$-RDP.
\end{theorem}

Another important property of RDP, and more generally DP, is that any further operations on the output of an RDP algorithm does not leak additional information, called the post-processing property.

\begin{theorem}[Post-Processing \cite{mironov2017renyi}]\label{thm:post-processing}
   Let $A: \mathcal{D} \rightarrow \mathcal{R}$ be $(\alpha, \epsilon(\alpha))$-RDP, and let $F: \mathcal{R} \rightarrow \mathcal{Z}$ be an arbitrary randomized mapping. Then $F \circ M$ is $(\alpha, \epsilon(\alpha))$-RDP.
\end{theorem} 

Another useful relaxation of the RDP definition is approximate RDP. 

\begin{definition}[Approximate RDP \citep{bun2016concentrated, zhu2022adaptive}]
\label{def:approximate-RDP}
We say a randomized algorithm $M$ is $\delta$-approximately $(\alpha, \epsilon_{M}(\alpha))$-RDP with order $\alpha \geq 1$, if for all neighboring dataset $D, D'$, there exist events $E$ (depending on $M(D)$) and $E'$ (depending on $M(D')$) such that $\Pr[E] \geq 1-\delta$ and $\Pr[E'] \geq 1-\delta$, and $\forall \alpha \geq 1$, we have 
\begin{align}
D_\alpha \left(M(D)|E~\|~ M\left(D^{\prime}\right)|E^{\prime} \right) \leq \epsilon_{M}(\alpha)
\end{align}
\end{definition}

Note that when $\delta=0$, then $0$-approximate $(\alpha, \epsilon(\alpha))$-RDP is simply $(\alpha, \epsilon(alpha))$-RDP. Finally, we can convert between $(\alpha, \epsilon(\alpha))$ and $(\epsilon, \delta)$-DP, which is shown below.

\begin{theorem}[Conversion from approximate RDP to Approximate DP \cite{zhu2022adaptive}]\label{thm:rdp-dp}
    If an algorithm $A$ satisfies $\delta_1$-approximate $(\alpha, \epsilon(\alpha))$-RDP, then it is $(\epsilon(\alpha) + \frac{\log(1 / \delta)}{\alpha -1}, \delta + \delta_1)$-DP for any $0 < \delta < 1$.
\end{theorem}

We use approximate RDP for a tighter measure of the privacy cost under composition. After we obtain the (approximate) RDP guarantee for the overall algorithm, we can then convert the privacy guarantee back into the standard DP definition (Theorem \ref{thm:rdp-dp}). 

Lastly, we introduce a fundamental DP mechanism, called the exponential mechanism \cite{mcsherry2007mechanism}, which the \texttt{FindBestK} is based on. Given some utility function over outputs, the exponential mechanism samples high-utility outputs with higher probability than low-utility outputs. 

\begin{definition}[Exponential Mechanism \cite{mcsherry2007mechanism}]
    Given a utility function $q: \mathcal{X}^* \times \mathcal{O} \rightarrow \mathcal{R}$ with $\ell_1$ sensitivity $\Delta(q) = \max_{D\sim D', o\in O} |q(D, o)-q(D', o)|$, the exponential mechanism $M$ has output distribution 
    \begin{equation*}
        \Pr[M(D) = o] \propto \exp\left (\frac{\epsilon q(D, o)}{2\Delta(q)} \right).
    \end{equation*}
    where $\propto$ elides the normalization factor.
\end{definition}

Given the above definition, one can show that the exponential mechanism satisfies $(\epsilon, 0)$-DP.

\begin{lemma}[\cite{mcsherry2007mechanism}]
    The exponential mechanism is $(\epsilon, 0)$-DP.
\end{lemma}

Since we want to compose the privacy loss of the exponential mechanism and report the loss in terms of $(\epsilon, \delta)$-DP, we use the following property to convert from $\epsilon$-DP to $(\alpha, \epsilon(\alpha))$-RDP. 

\begin{theorem}[\cite{bun2016concentrated}]\label{thm:dp_rdp}
    The $M$ is $\epsilon$-DP then it is $(\alpha, \epsilon_{\text{EM}}(\alpha))$-RDP where
    \begin{equation*}
        \epsilon_{\text{EM}}(\alpha) := \min \left (\frac{\alpha}{2}\epsilon^2, \frac{1}{\alpha-1} \log \left (\frac{\sinh(\alpha\epsilon) - \sinh((\alpha-1)\epsilon)}{\sinh(\epsilon)} \right ) \right).
    \end{equation*}
\end{theorem}

\subsection{Privacy Analysis}

Now we will prove that \method{} satisfies $(\epsilon, \delta)$-DP. First, we obtain a privacy guarantee for \texttt{TopKWithPTR} and \texttt{FindBestK}.

\begin{theorem}\label{thm:ptr_dp}
    \texttt{TopKWithPTR} (Algorithm \ref{alg:ptr}) is $\delta$-approximate $\frac{\alpha}{2 \sigma^2}$-RDP.
\end{theorem}

\begin{proof}
The privacy analysis mostly follows from \citet{zhu2022adaptive}. Releasing the noisy threshold $\widehat d_k$ is $\frac{\alpha}{2 \sigma^2}$-RDP. 

    If $d_k > 2$, then releasing the exact top-$k$ tokens has no privacy cost, as its local sensitivity is 0. 

    If $d_k \le 2$, then if $\widehat d_k \le 2$, the program terminates and there's no privacy cost. 

    If $d_k \le 2$, the failure probability
    \begin{align*}
        \Pr[\widehat d_k > 2] 
        &= \Pr[ \max(2, d_k) + \mathcal{N}(0, 4\sigma^2) - \Phi(1-\delta; 0, 2\sigma) > 2 ] \\
        &= \Pr[ 2 + \mathcal{N}(0, 4\sigma^2) - \Phi(1-\delta; 0, 2\sigma) > 2 ] \\
        &= \Pr[ \mathcal{N}(0, 4\sigma^2) - \Phi(1-\delta; 0, 2\sigma) > 0] \\
        &= \delta
    \end{align*}
\end{proof}

\begin{theorem}\label{thm:findbestk_dp}
    \texttt{FindBestK} (Algorithm \ref{alg:rnm-find-k}) satisfies $(\alpha, \epsilon_{\text{EM}}(\alpha))$-RDP.
\end{theorem}

\begin{proof}
    Note that adding Gumbel noise to each output's utility and releasing the output with the highest noisy utility score is equivalent to using the exponential mechanism \cite{durfee2019practical}. Hence, \texttt{FindBestK} is an EM where the utility function is $d_k = \mathbf{H}_{(k)}-\mathbf{H}_{(k+1)}$ with the sensitivity $\Delta(d_k) = 2$ (this arguement derives from \citet{zhu2022adaptive}). Therefore, the privacy guarantee follows from Theorem \ref{thm:dp_rdp}.
\end{proof}

\begin{theorem}
    \method{} (Algorithm \ref{alg:dpkwrag}) satisfies $(\epsilon, \delta)$-DP. 
\end{theorem}

\begin{proof}
    Let $D, D'$ be two adjacent datasets the differ by one document. Hence, after retrieving the $N$ most relevant documents for a query $\mathbf{x}$, $R(\mathbf{x}, D; N)$ and $R(\mathbf{x}, D'; N)$ differ by at most one document. Suppose it is the $i$-th document. More precisely, $D^{\mathbf{x}}_{j} = D'^{\mathbf{x}}_j$ $\forall j \neq i$ and $D^{\mathbf{x}}_{i} \neq D'^{\mathbf{x}}_i$. Then we obtain the corresponding histograms $\mathbf{H}$ and $\mathbf{H}'$, which differ by one response $\mathbf{y}_i \neq \mathbf{y}'_i$. Hence by Theorem \ref{thm:findbestk_dp} $\hat{k}$ is $(\alpha, \epsilon_{\text{EM}}(\alpha))$-RDP. Then using $\hat{k}$ we have $w_{1}, ..., w_{\hat{k}}$, which are $\delta_1$-approximate $\frac{\alpha}{2 \sigma^2}$-RDP by Theorem \ref{thm:ptr_dp}. Then by composition (Theorem \ref{thm:composition}) the total privacy loss is $\delta_1$-approximate $(\alpha, \epsilon_{\text{EM}}(\alpha) + \frac{\alpha}{2 \sigma^2})$-RDP. Moreover, because $w_{1}, ..., w_{\hat{k}}$ are RDP, any further operations on them do not leak additional privacy by post-processing (Theorem \ref{thm:post-processing}). Hence, using $w_{1}, ..., w_{\hat{k}}$ for obtaining the final output $\mathbf{y} \gets F(\mathbf{x}, \{w_i\}_{i=1}^{\hat{k}})$ is DP. Finally we convert the final privacy loss back to $(\epsilon, \delta)$-DP using Theorem \ref{thm:rdp-dp}.
\end{proof}